\RequirePackage{amsmath}
\documentclass[runningheads]{llncs}
\usepackage{amssymb}
\usepackage{verbatim}
\usepackage{epic,gastex}
\usepackage{array}
\usepackage{algorithmicx,algorithm,algpseudocode}
\usepackage{xcolor}
\usepackage[lowtilde]{url}
\usepackage{graphicx}

\DeclareMathOperator{\rt}{rt}

\DeclareSymbolFont{rsfscript}{OMS}{rsfs}{m}{n}
\DeclareSymbolFontAlphabet{\mathrsfs}{rsfscript}

\begin{document}

\title{On the Number of\\Synchronizing Colorings of Digraphs}

\author{Vladimir V. Gusev\inst{1}\thanks{Supported by the Communaut{\'e} fran\c{c}aise de Belgique
- Actions de Recherche Concert{\'e}es, by the Belgian Programme on
Interuniversity Attraction Poles, by the Russian foundation for basic research(grant 13-01-00852), Ministry of Education and Science of the Russian Federation (project no.\ 1.1999.2014/K), Presidential Program for Young Researchers(grant MK-3160.2014.1) and the Competitiveness Program of Ural Federal University.}
\and
Marek Szyku{\l}a\inst{2}\thanks{Supported in part by Polish NCN grant DEC-2013/09/N/ST6/01194.}}

\authorrunning{V. V. Gusev, M. Szyku{\l}a}

\tocauthor{V. V. Gusev,  M. Szyku{\l}a}

\institute{ICTEAM Institute, Universit{\'e} catholique de Louvain, Belgium,\\
Ural Federal University, Russia\\
\email{vl.gusev@gmail.com}
\and Institute of Computer Science,\\
University of Wroc{\l}aw, Wroc{\l}aw, Poland\\
\email{msz@cs.uni.wroc.pl}}

\maketitle

\begin{abstract}
We deal with $k$-out-regular directed multigraphs with loops (called simply \emph{digraphs}).
The edges of such a digraph can be colored by elements of some fixed $k$-element set in such a way that outgoing edges of every vertex have different colors. Such a coloring corresponds naturally to an automaton.
The road coloring theorem states that every primitive digraph has a synchronizing coloring.

In the present paper we study how many synchronizing colorings can exist for a digraph with $n$ vertices.
We performed an extensive experimental investigation of digraphs with small number of vertices. This was done by using our dedicated algorithm exhaustively enumerating all small digraphs.
We also present a series of digraphs whose fraction of synchronizing colorings is equal to $1-1/k^d$, for every $d \ge 1$ and the number of vertices large enough.

On the basis of our results we state several conjectures and open problems. In particular, we conjecture that $1-1/k$ is the smallest possible fraction of synchronizing colorings, except for a single exceptional example on 6 vertices for $k=2$.
\end{abstract}


\section{Introduction}

Throughout the paper we deal with directed multigraphs $\mathcal{G} = \langle V,E\rangle$ of a fixed out-degree $k$ with loops, where $V$ is a finite set of $n$ \emph{vertices} and $E$ is a finite multiset of \emph{edges}. For each $v \in V$ there are exactly $k$ outgoing edges $(v,w) \in E$. These are simply called \emph{digraphs} throughout the paper.
For every vertex $v \in V$ of a digraph $\mathcal{G}$, the $k$ outgoing edges can be colored differently by one of the $k$ colors from a finite set $\Sigma$, giving raise to a \emph{deterministic finite (semi)automaton} $\mathrsfs{A} = \langle V,\Sigma,\delta\rangle$ with the set of states $V$, the alphabet $\Sigma$, and the transition function $\delta$, where $\delta(v,a) = w$ whenever an edge $(v,w) \in E$ was colored by $a$. Every such automaton $\mathrsfs{A}$ is called a \emph{coloring} of the digraph $\mathcal{G}$.
Thus we identify automata with colorings of their underlying digraphs.
We extend the transition function $\delta\colon V \times \Sigma \to V$ to $\delta\colon 2^V \times \Sigma^* \to 2^V$ on subsets and words in the natural way.
For $\delta(S,w)$, where $S \subseteq V$ and $w \in \Sigma^*$, we also write shortly $S w$.

Automaton $\mathrsfs{A}$ is called \emph{synchronizing}
if there exist a word $w$ and a state $p$ such that for every state $q \in Q$ we have $\delta(q,w)=p$. Such a word $w$ is called \emph{reset} (or \emph{synchronizing}) word for $\mathrsfs{A}$.
The length of the shortest synchronizing word of $\mathrsfs{A}$ is called \emph{reset threshold} and is denoted by $\rt(\mathrsfs{A})$.
Recent surveys of the theory of synchronizing automata may be found in~\cite{KariVolkov2013Handbook,Volkov2008Survey}.

Call a digraph \emph{primitive} (or \emph{aperiodic}) if it is strongly connected and the $\gcd$ of all its cycles is equal to 1.
It is easy to show that an underlying digraph of a synchronizing automaton is primitive. In 1977 Adler, Goodwyn and Weiss conjectured~\cite{AGW1977} that every primitive digraph has a synchronizing coloring. This conjecture became widely known as the \emph{road coloring} problem. It was arguably one of the most important conjectures in automata theory until it was finally proved by Trahtman in 2007~\cite{Tr2009RoadColoring}. One of the goals of the present paper is to find the right quantitative formulation of the road coloring theorem. 

Another part of our motivation comes from the algorithmic issues related to the road coloring problem. How to find a synchronizing coloring of a given digraph?
A non-trivial algorithm working in time $O(kn^2)$ is known for this task~\cite{BealPerrin2014AQuadraticAlgorithmForRoadColoring}. On the other hand, M.-P. B{\'e}al suggested during her talk at CANT 2012 that a random sampling of colorings in a search for a synchronizing one may lead to a simple and practically effective algorithm for the problem. Since one can check whether a coloring is synchronizing in $O(kn^2)$ time, it remains to show that a random coloring is synchronizing with high probability. In our research we were partially motivated by this observation.

There are other computational problems related to the synchronizing colorings of digraphs, such as deciding existence of a synchronizing coloring for a fixed reset word~\cite{VorelRoman2015ComplexityOfRoadColoringWithPrescribedResetWords}, or for a fixed reset threshold~\cite{Roman2012P-NPThresholdForSynchronizingRoadColoring}.
Also, several open problems concerning synchronizing automata and the road coloring problem have been stated by M.V. Volkov~\cite{Volkov2008OpenProblemsOnSynchronizingAutomata}.



For a given $k$-out-regular digraph $\mathcal{G}$ with $n$ vertices,
the $\emph{synchronizing ratio}$ is the number of synchronizing colorings to the number $(k!)^n$ of all possible colorings.
Note that we distinguish edges of $\mathcal{G}$, i.e. two colorings are the same if all edges have the same color. Therefore, there is always exactly $(k!)^n$ different colorings of $\mathcal{G}$.
A digraph $\mathcal{G}$ is $\emph{totally synchronizing}$ if its synchronizing ratio of $\mathcal{G}$ is equal to 1.

In this paper we perform an experimental and theoretical study on the synchronizing ratio of digraphs. Our main contributions are as follows:
\begin{enumerate}
\item We developed an efficient algorithm for enumerating and checking synchronizing ratios of nonisomorphic digraphs.
\item Using the algorithm, we performed extensive experiments revealing various phenomena concerning the synchronizing ratio.
These provide evidence to state several conjectures and form a basis for further investigation.
\item We found out that for small $n$ and $k$ there are no primitive strongly connected digraphs with synchronizing ratio less than $1-1/k$, except for a single particular example for $n=6$ and $k=2$.
\item We constructed digraphs with synchronizing ratio $1-1/k^d$, for every $d \ge 1$ and $n\ge 3d$. This shows that there are many examples with different synchronizing ratio in the range $[1-1/k,1]$.
\end{enumerate}

\section{General Statements}

A strongly connected component $S$ of a digraph $\mathcal{G} = \langle V,E\rangle$ is called a \emph{sink component} if there are no edges going from $S$ to $V \setminus S$.
It is \emph{reachable} if for any vertex $v \in V$ there is a directed path from $v$ to a vertex in $S$.

\begin{proposition}\label{pro:nonsc}
If a digraph $\mathcal{G}$ has a synchronizing coloring then it has a unique reachable sink component $S$.
Furthermore, the synchronizing ratio of $\mathcal{G}$ is equal to the synchronizing ratio of the digraph induced by $S$.
\end{proposition}
\begin{proof}
The proof of the first statement belongs to folklore.
It is not hard to see that an arbitrary coloring $\mathrsfs{A}$ of digraph $\mathcal{G}$ is synchronizing if and only if the subautomaton $\mathrsfs{A}'$ induced by the sink component $S$ is synchronizing.
Therefore, the set of all colorings of $\mathcal{G}$ can be divided into groups of equal size, each group containing the colorings with the same induced subcoloring of $S$.
Since colorings from each group are altogether synchronizing or non-synchronizing, we obtain that the
synchronizing ratio of $\mathcal{G}$ is equal to the synchronizing ratio of the digraph induced by $S$.
\qed
\end{proof}

Since a one-vertex digraph is totally synchronizing we have the following corollary:
\begin{corollary}
A digraph with a sink state is either totally synchronizing or none of its colorings is synchronizing.
\end{corollary}

Due to Proposition~\ref{pro:nonsc} the study of synchronizing ratios and totally synchronizing digraphs can be reduced to the case of strongly connected digraphs.

Surprisingly, the underlying digraphs of several automata presented in the literature appear to be totally synchronizing. One important example of such a digraph is well known to the community, see for example~\cite{Volkov2008OpenProblemsOnSynchronizingAutomata}. Recall that the \v{C}ern\'{y} automaton $\mathrsfs{C}_n$ (\cite{Cerny1964}) can be defined as $\langle \{0,\ldots,n-1\},\{a,b\},\delta \rangle$, where $\delta(i,a)=i+1$ for $i<n-1$, $\delta(n-1,a)=0$, $\delta(n-1,b)=0$, and $\delta(i,b)=i$ for $i<n-1$. The proof of the following folklore result has not yet appeared in the literature.

\begin{proposition}
The underlying digraph of $\mathrsfs{C}_n$ is totally synchronizing.
\end{proposition}
\begin{proof}
Let $\mathrsfs{C}'_n$ be an arbitrary coloring of the underlying digraph of $\mathrsfs{C}_n$.
It is well known that an automaton is synchronizing if and only if every pair of states $i,j$ is synchronizing, i.e. there is a word $w$ such that $i w = j w$ (see~\cite{Cerny1964}, or \cite[Proposition~2.1]{Volkov2008Survey}).

We will show that any pair of states $(i, j)$ of $\mathrsfs{C}'_n$ satisfy this condition.
Let $d(i,j)$ be the length of the shortest path from $i$ to $j$.
We will proceed by induction on $d(i,j)$.
Consider a pair $(i,j)$; without loss of generality we may assume that $d(i,j) \leq d(j,i)$.

If $j = n-1$ then let $y$ be the letter on the edge from $i$ to $i+1$.
We apply $y$ so $(i y, j y) = (i+1, 0)$, and $d(i, j) = d(i+1, 0)$.

Consider the case $j \neq n-1$.
Let $x$ be the letter on the loop on the state $j$.
If $i < n-1$ then let $y$ be the letter on the edge from $i$ to $i+1$; otherwise let $y = x$.
If $x=y$ then $d(i x,j x) = d(i y, j) < d(i,j)$.
Otherwise we apply the letter $y$, and in the same manner consider the pair $(i y, j y) = (i+1, j+1)$.
Note that $d(i,j) = d(i+1 , j+1)$.
Following in this way, after at most $n-1-i$ steps, we will reach a pair $(n-1, k)$.
For $(n-1, k)$ we choose $y = x$, we obtain $d((n-1) x, k x) = d(0, k) < d(n-1, k)$.
\qed
\end{proof}

Underlying digraphs of many other automata that appeared in literature are also totally synchronizing.
In a similar fashion one can show that the underlying digraphs of the series of slowly synchronizing automata (see~\cite{AGV2013,KS2014SynchronizingAutomataWithLargeResetLengths}) are totally synchronizing. Also, almost all presented examples of automata with two cycle lengths have this property~\cite{GusevPribavkina2014ResetThresholdsOfAutomataWithTwoCycleLengths}.

For the sake of completeness we mention the following notions from related topics.
A word $w$ is called \emph{totally synchronizing} if $w$ is a reset word for any coloring of totally synchronizing digraph $\mathcal{G}$. See~\cite{Cardoso2014PhD} for an analysis of totally synchronizing digraphs and words in some special classes of digraphs. A word over an alphabet $\Sigma$ is called \emph{$n$-synchronizing} if it is a reset word for all synchronizing automata with $n + 1$ states over the alphabet $\Sigma$. See~\cite{Cherubini2007SynchronizingAndCollapsingWords} for the introduction to the topic.

\section{Experimental Investigation of Digraphs}

We performed a series of experiments to reveal some properties of the synchronizing ratio of digraphs. These include both exhaustive enumeration of small digraphs and larger random digraphs.
We are interested mostly in primitive strongly connected digraphs (cf.~Proposition~\ref{pro:nonsc}).
In the case of exhaustive enumeration we checked the synchronizing ratio of all nonisomorphic $k$-out-regular digraphs with a given $n$ vertices.

\subsection{Algorithms}

To check as many cases as possible and obtain a large data set, we needed to design and implement our algorithms carefully.
This is especially important during the exhaustive search, since the number of digraphs grows very fast with $n$ and $k$.
Here we briefly describe our algorithms, skipping numerous technical improvements and tricks in the implementation.
Some of our ideas are based on~\cite{KS2013GeneratingSmallAutomata}, where the \v{C}ern\'{y} conjecture was verified by an exhaustive enumeration for all binary automata up to $n\le 11$ states.

To determine the synchronizing ratio of a digraph, we can just enumerate all its colorings and count the synchronizing ones.
Checking whether a coloring (automaton) is synchronizing can be easily done in $O(kn^2)$ time \cite{Cerny1964,Ep1990}.
Note that in many cases, some colorings give rise to the same particular automaton (e.g.~if there are are two or more parallel edges $(v,w)$ then we can permute its colors obtaining the same automaton).
Also, every coloring has $k!$ equivalent colorings obtained only by permuting the colors. Using these facts we could greatly reduce the total number of really checked colorings for synchronization.

Checking whether a digraph is strongly connected and the $\gcd$ of its cycles is 1 can be effectively done in $O(kn)$ time basing on the algorithms from~\cite{Tarjan1972} and \cite{JS1996GraphTheoreticAnalysisOfMarkovChains}, respectively.

Now, computing the synchronizing ratios of a set of random digraphs follows easily, and we can proceed this in parallel on a grid.
However, in an exhaustive enumeration, the number of digraphs grows very fast in terms of $n$ and $k$ (see~Table~\ref{tab:exhaustive_results_2}), and the main problem was to deal with it.

\begin{algorithm}
\caption{Exhaustive checking of digraphs.}
\label{alg:exhaustive}
\begin{algorithmic}[1]
\Require{$n$ -- the number of vertices (states)}
\Require{$k$ -- the out-degree (size of the alphabet)}
\State $G_n \gets$ the set of all simple graphs with $n$ vertices.
\ForAll{simple graphs $\mathcal{G} \in G_n$}\Comment{In parallel}
  \State $\mathit{CanSet} \gets$ \Call{EmptySet}{}
  \ForAll{digraphs $\mathcal{D}_{n,k}$ with underlying graph $\mathcal{G}$}\Comment{ Orient and multiply the edges of $\mathcal{G}$, and add loops}
    \If{$\mathcal{D}_{n,k}$ is primitive}
      \State $\mathcal{R}_{n,k} \gets$ the canonical representation of $\mathcal{D}_{n,k}$.
      \If{$\mathcal{R}_{n,k} \not\in \mathit{CanSet}$}
        \State $\mathit{CanSet}$.\Call{insert}{$\mathcal{R}_{n,k}$}
        \State Count synchronizing colorings of $\mathcal{R}_{n,k}$.
      \EndIf
    \EndIf
  \EndFor
\EndFor
\end{algorithmic}
\end{algorithm}

Our algorithm for exhaustive checking of digraphs is summarized in~Algorithm~\ref{alg:exhaustive}.
First, in line~1, we generate all nonisomorphic simple graphs with $n$ vertices. A \emph{simple graph} is a graph with undirected edges joining two distinct vertices. This can be done effectively by the algorithm from \cite{McKayPiperno2014PracticalGraphIsomorphismII}, implemented in package \texttt{nauty}.
Now, we can process each such a simple graph in parallel. In line~4, for every simple graph $\mathcal{G}$ we orient and multiply its edges so that there are at most $k$ outgoing edges for each vertex. Then we interpret the missing edges as loops.
Clearly, an isomorphic copy of every digraph can be obtained in this way from its underlying simple graph, and the digraphs obtained from two nonisomorphic simple graphs are also nonisomorphic. We can, however, obtain isomorphic digraphs from the same simple graph.
In line~5 we skip non-strongly connected and non-primitive digraphs.
In line~6 we compute the \emph{canonical representation} of a generated digraph $\mathcal{D}_{n,k}$; this is the lexicographically minimal representation among all digraphs isomorphic to $\mathcal{D}_{n,k}$ (cf.~\cite{McKayPiperno2014PracticalGraphIsomorphismII}).
To skip isomorphic copies obtained from the same simple graph, in line~3 we introduce the set $\mathit{CanSet}$ of canonical representations of generated digraphs. Then in line~7, we check if an isomorphic copy of the digraph $\mathcal{D}_{n,k}$ was already generated; if not, in line~8 we insert it to the set. The set $\mathit{CanSet}$ can be effectively implemented as a radix trie, allowing to perform both membership test and insertion in linear time, and providing some compression (which is also important in view of the number of generated digraphs).
Finally, we can count synchronizing colorings of the generated digraph (line~9).

\subsection{Experimental Results from Exhaustive Enumeration}

\begin{table}
\caption{The minimum, average, and standard deviation of the number of synchronizing colorings of all strongly connected aperiodic $k$-out-regular digraphs with $n$ vertices}
\label{tab:exhaustive_results_1}
\begin{center}
\begin{tabular}{|c|c|r|r|r|r|r|}
  \hline
  $k$&$n$&\multicolumn{1}{c|}{\ Min\ }&\multicolumn{1}{c|}{\ Min ratio\ }&\multicolumn{1}{c|}{\ Avg\ }& \multicolumn{1}{c|}{\ Avg ratio\ } &\multicolumn{1}{c|}{\ Std dev\ } \\ \hline\hline
  2 & 2  &  2         & 0.5       & 3               & 0.750           & 1.000  \\ \hline  
  2 & 3  &  4         & 0.5       & 6.833           & 0.854           & 1.280  \\ \hline
  2 & 4  &  8         & 0.5       & 14.640          & 0.915           & 2.243  \\ \hline  
  2 & 5  &  16        & 0.5       & 30.987          & 0.968           & 2.146  \\ \hline  
  2 & 6  &  30        & 0.469     & 63.139          & 0.986           & 2.381  \\ \hline
  2 & 7  &  64        & 0.5       & 127.365         & 0.995           & 2.033  \\ \hline  
  2 & 8  &  128       & 0.5       & 255.483         & 0.998           & 1.866  \\ \hline  
  2 & 9  &  256       & 0.5       & 511.563         & 0.999           & 1.617  \\ \hline  
  2 & 10 &  512       & 0.5       & 1,023.607       & $\approx 1.000$ & 1.468  \\ \hline
  3 & 2  &   24       & 0.667     & 31.2            & 0.867           & 5.879  \\ \hline 
  3 & 3  &  144       & 0.667     & 208.800         & 0.967           & 14.163  \\ \hline 
  3 & 4  &  864       & 0.667     & 1,284.987       & 0.991           & 36.346  \\ \hline 
  3 & 5  & 5,184      & 0.667     & 7,765.775       & 0.999           & 50.091  \\ \hline 
  3 & 6  & 31,104     & 0.667     & 46,643.953      & $\approx 1.000$ & 78.679  \\ \hline 
  3 & 7  & 186,624    & 0.667     & 279,921.191     & $\approx 1.000$ & 108.167  \\ \hline 
  4 & 2  & 432        & 0.75      & 533.333         & 0.926           & 61.738  \\ \hline  
  4 & 3  & 10,368     & 0.75      & 13,704.874      & 0.991           & 367.767  \\ \hline
  4 & 4  & 248,832    & 0.75      & 331,421.072     & 0.999           & 2,233.171  \\ \hline
  4 & 5  & 5,971,968  & 0.75      & 7,961,941.49    & $\approx 1.000$ & 7,104.373  \\ \hline
  5 & 2  & 11,520     & 0.75      & 13,782.857      & 0.957           & 1,048.941  \\ \hline  
  5 & 3  & 1,382,400  & 0.75      & 1,723,468.312   & 0.997           & 720,951.433  \\ \hline
  5 & 4  & 165,888,000& 0.75      & 207,324,196.845 & $\approx 1.000$ & 412,162.118  \\ \hline
\end{tabular}
\end{center}
\end{table}

The algorithms in C++ and compiled with GCC 4.8.1.
The computations were performed in parallel on a small grid consisted of computers with 8 processors Quad-Core AMD Opteron(tm) 8350 (2 GHz) and 64GB of RAM.

We were able to check all 2-out-regular digraphs with up to 10 vertices, 3-out-regular up to 7 states, 4-out-regular up to 5 states, and 5-out-regular up to 4 states.
In the case of $k=2$-out-regular digraphs with $n=10$ states, the total processor time was more than 60 days (about 1 day of parallelized computation).
The case of $k=3$ and $n=7$ took even more, about 72 days; the total number of colorings was $\sim 7 \times 10^{14}$, but, thanks to optimization, we required to check only $\sim 10^{13}$ automata.

The results concerning synchronizing ratios are summarized in Table~\ref{tab:exhaustive_results_1}.
In Table~\ref{tab:exhaustive_results_2} we present the exact number of strongly connected aperiodic digraphs, and totally synchronizing digraphs.
We observe that the fraction of totally synchronizing digraphs within the class of strongly connected aperiodic digraphs is growing.

\begin{table}
\caption{The number of nonisomorphic strongly connected aperiodic digraphs, the number of totally synchronizing digraphs, and their fraction}
\label{tab:exhaustive_results_2}
\begin{center}
\begin{tabular}{|c|c|r|r|r|}
  \hline                       
  $k$ & $n$ &\multicolumn{1}{c|}{\ S.c. aperiodic\ }&\multicolumn{1}{c|}{\ Totally synchronizing\ }&\multicolumn{1}{c|}{\ Fraction\ }\\ \hline\hline
  2 & 2 & 2 & 1 & 0.500 \\ \hline  
  2 & 3 & 12 & 6 & 0.500 \\ \hline  
  2 & 4 & 100 & 66 & 0.660 \\ \hline  
  2 & 5 & 1220 & 890 & 0.729 \\ \hline  
  2 & 6 & 19,064 & 14,973 & 0.785\\ \hline  
  2 & 7 & 361,157 & 296,303 & 0.82\\ \hline  
  2 & 8 & 8,001,589 & 6,754,895 & 0.844\\ \hline  
  2 & 9 & 202,635,930 & 174,246,295 & 0.860\\ \hline  
  2 & 10 & 5,765,318,112 & 5,026,305,042 & 0.872\\ \hline  
  3 & 2 & 5 & 3 & 0.600 \\ \hline     
  3 & 3 & 85 & 63 & 0.741 \\ \hline  
  3 & 4 & 3,148 & 2,672 & 0.849 \\ \hline  
  3 & 5 & 199,489 & 182,326 & 0.914 \\ \hline  
  3 & 6 & 19,059,581 & 18,006,297 & 0.945 \\ \hline  
  3 & 7 & 2,537,475,117 & 2,443,850,969 & 0.963 \\ \hline    
  4 & 2 & 9 & 6 & 0.666 \\ \hline
  4 & 3 & 357 & 302 & 0.846 \\ \hline
  4 & 4 & 39,680 & 36,762 & 0.926 \\ \hline
  4 & 5 & 9,089,413 & 8,779,342 & 0.966 \\ \hline
  5 & 2 & 14 & 10 & 0.714 \\ \hline
  5 & 3 & 1,102 & 990 & 0.898 \\ \hline
  5 & 4 & 304,082 & 291,530 & 0.959 \\ \hline
\end{tabular}
\end{center}
\end{table}

%


In Table~\ref{tab:gaps}, for $k=2$ and $n=8$, we present the number of nonisomorphic digraphs with particular numbers of synchronizing colorings.
Interestingly, there are several graphs in the distribution, and the gaps grow for smaller number of synchronizing colorings.
This picture is similar for the other values of $n$ and $k$ that we checked. The number of gaps seems to grow with $n$ and $k$.

\begin{table}
\caption{The number of nonisomorphic digraphs with the given number of synchronizing colorings for $k=2$ and $n=8$}
\label{tab:gaps}
\begin{center}
\begin{tabular}{|l||c|c|c|c|c|c|c|c|c|c|c|c|c|c|c|c|c|} \hline
\# synch. col. & 128 & 130 & \ldots & 158 & 160 & 162 & \ldots & 174 & 176 & 178 & 180 & 182 & 184 & 186 & 188 & 190  & 192 \\ \hline
\# digraphs    & 72  & 0   & \ldots & 0   & 24  & 0   & \ldots & 0   & 1   & 0   & 0   & 0   & 5   & 0   & 0   & 1    & 813 \\ \hline\hline
\# synch. col. & 194 & 196 & 198    & 200 & 202 & 204 & 206    & 208 & 210 & 212 & 214 & 216 & 218 & 220 & 222 & 224  & 226 \\ \hline
\# digraphs    & 0   & 1   & 1      & 12  & 1   & 1   & 6      & 202 & 0   & 2   & 1   & 134 & 4   & 22  & 14  & 4,022& 60  \\ \hline
\end{tabular}
\end{center}
\begin{center}
\begin{tabular}{|l||c|c|c|c|c|c|c|c|} \hline
\# synch. col. & 228   & 230  & 232   & 234    & 236    & 238    & 240      & 242 \\ \hline
\# digraphs    & 73    & 170  & 852   & 179    & 1,226  & 610    & 21,933   & 699 \\ \hline \hline
\# synch. col. & 244   & 246  & 248   & 250    & 252    & 254    & \multicolumn{2}{c|}{256} \\ \hline
\# digraphs    & 4,523 & 3,171& 44,230& 27,438 & 310,400& 825,791& \multicolumn{2}{c|}{6,754,895} \\ \hline
\end{tabular}
\end{center}
\end{table}

\subsection{Experiments on Random Digraphs}

To deal also with larger digraphs, we performed additional experiments with random digraphs.
We used the uniform model of a random digraph, that is, for every outgoing edge from a vertex $v$ we choose the destination vertex uniformly at random and independently from the other choices.

For $k=2$ we checked $1000,000$ digraphs for every $n=4,\ldots,15$, and $100,000$ for $n=16,\ldots,27$.
Since the number of possible colorings grow very fast with $k$, for $k=3$ we tested $n=4,\ldots,12$, and for $k=4$ we tested only $n=4,\ldots,8$.
We additionally checked the same numbers of random of digraphs in the class of strongly connected aperiodic digraphs, within the same range of $n$ and $k$.

The results from random tests for larger $n$ show the same patterns as observed in those from exhaustive search.
Figure~\ref{fig:totally} shows the fraction of totally synchronizing digraphs in random samples of strongly connected aperiodic digraphs.
The picture is very similar in the class of all digraphs.

\begin{figure}
\centering
\includegraphics[width=\textwidth]{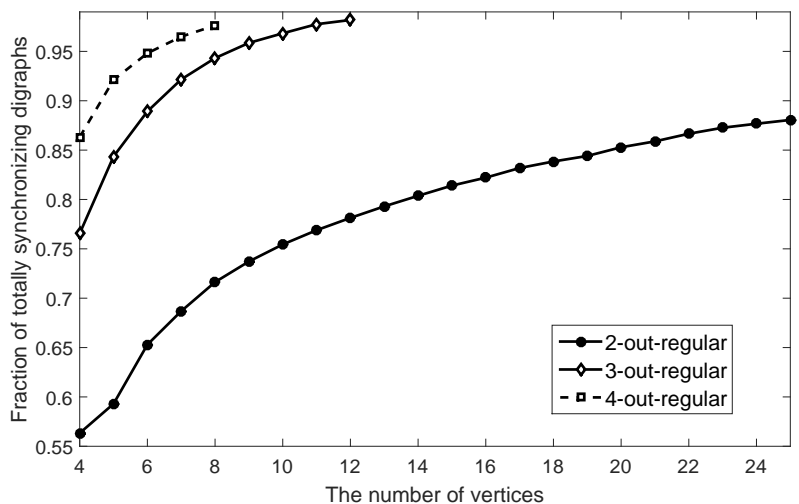}
\caption{The fraction of totally synchronizing digraphs in the class of strongly connected and aperiodic digraphs}
\label{fig:totally}
\end{figure}

\section{Digraphs with Specific Synchronizing Ratios}

In this section we present different examples of digraphs with particular values of the synchronizing ratio.
According to our computational experiments the smallest possible value of the synchronizing ratio among all digraphs is equal to $\frac{30}{64}$. This value is achieved by the digraph $\mathcal{G}_{30}$ (Figure~\ref{fig:30}).
By direct computation one can verify that only 30 colorings of $\mathcal{G}_{30}$ are synchronizing.

\begin{proposition}
There is a 2-out-regular digraph with 6 states and the synchronizing ratio $\frac{30}{64}$.
\end{proposition}

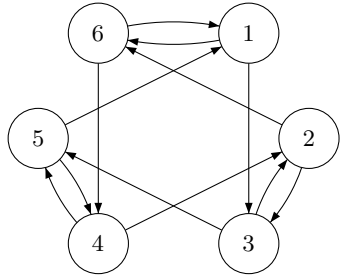
\begin{figure}[ht]
\begin{center}
\unitlength .5mm
\begin{picture}(80,76)(0,-76)
\gasset{Nw=16,Nh=16,Nmr=8} 
\node(n1)(60.0,-16.0){$1$}
\node(n2)(76.0,-44.0){$2$} 
\node(n3)(60.0,-72.0){$3$}
\node(n4)(20.0,-72.0){$4$}
\node(n5)(4.0,-44.0){$5$}
\node(n6)(20.0,-16.0){$6$}
\drawedge[ELdist=1.7](n1,n3){}
\drawedge[curvedepth=3](n1,n6){}
\drawedge[ELdist=1.7](n2,n6){}
\drawedge[curvedepth=3](n2,n3){}
\drawedge[curvedepth=3](n3,n2){}
\drawedge[ELdist=1.7](n3,n5){}
\drawedge[ELdist=1.7](n4,n2){}
\drawedge[curvedepth=3](n4,n5){}
\drawedge[curvedepth=3](n5,n4){}
\drawedge[ELdist=1.7](n5,n1){}
\drawedge[ELdist=1.7](n6,n4){}
\drawedge[curvedepth=3](n6,n1){}
\end{picture}
\end{center}
\caption{Digraph $\mathcal{G}_{30}$.}
\label{fig:30}
\end{figure}
The exceptional example $\mathcal{G}_{30}$ seems to be unique.
We did not find any other digraph with this particular value of the synchronizing ratio.
Furthermore, according to our computational experiments the smallest value of the synchronizing ratio among all other $k$-out-regular digraphs seems to be equal to $\frac{k-1}{k}$. There are many examples that reach this bound, and
in the following theorem we construct a series of digraphs with this property.

\begin{theorem}
For every $n > 3$ there is a $k$-out-regular digraph with $n$ vertices and the synchronizing ratio $\frac{k-1}{k}$.
\end{theorem}
\begin{proof}
We will define the digraph $\mathcal{G}_{n,k}$ as follows; see Figure~\ref{fig:half}.
The set of vertices $V$ is $\{0,1, \ldots n-1\}$.
The edges $(0,1)$ and $(1,2)$ are of multiplicity 1, the edges $(1,1)$ and $(0,2)$ are of multiplicity $k-1$, and the edges $(2,3),(3,4),\ldots,(i,i+1),\ldots,(n-1, 0)$ are of multiplicity $k$.

Consider now an arbitrary coloring of the digraph $\mathcal{G}_{n,k}$.
Let $x$ be the letter on the edge $(0,1)$ and $y$ be the letter on the edge $(1,2)$. If $x=y$ then every letter acts as a permutation, and so the automaton is not synchronizing.
If $x \neq y$ then $x^{n-1}$ is a reset word for the given coloring.
Hence, a coloring is synchronizing if and only if $x=y$. Therefore, the synchronizing ratio of $\mathcal{G}_{n,k}$ is equal to $\frac{k-1}{k}$.
\qed
\end{proof}

\begin{figure}[th]
\begin{center}
\unitlength .5mm
\begin{picture}(80,90)(0,-90)
\gasset{Nw=16,Nh=16,Nmr=8,loopdiam=10}
\node(n1)(40.0,-16.0){1}
\node(n)(4.0,-40.0){$0$}
\node(n2)(76.0,-40.0){2}
\node(n3)(72.0,-72.0){3}
\node(n-1)(8.0,-72.0){$n{-}1$}
\node[Nframe=n](ndots)(40,-88){$\dots$}
\drawedge(n,n1){}
\drawedge(n,n2){}
\drawedge[curvedepth=3](n2,n3){}
\drawedge[curvedepth=-3](n2,n3){}
\drawedge[ELdist=1.7](n1,n2){}
\drawedge[curvedepth=3](n-1,n){}
\drawedge[curvedepth=-3](n-1,n){}
\drawedge[curvedepth=3](n3,ndots){}
\drawedge[curvedepth=-3](n3,ndots){}
\drawedge[curvedepth=3](ndots,n-1){}
\drawedge[curvedepth=-3](ndots,n-1){}
\drawloop[ELdist=1.5,loopangle=-90](n1){}
\end{picture}
\begin{picture}(80,78)(-20,-90)
\gasset{Nw=16,Nh=16,Nmr=8,loopdiam=10}
\node(n1)(40.0,-16.0){1}
\node(n)(4.0,-40.0){$0$}
\node(n2)(76.0,-40.0){2}
\node(n3)(72.0,-72.0){3}
\node(n-1)(8.0,-72.0){$n{-}1$}
\node[Nframe=n](ndots)(40,-88){$\dots$}
\drawedge[ELdist=2.0](n,n1){}
\drawedge[ELdist=2.0, curvedepth=3](n,n2){}
\drawedge[ELdist=2.0, curvedepth=-3](n,n2){}
\drawedge[curvedepth=3](n2,n3){}
\drawedge(n2,n3){}
\drawedge[curvedepth=-3](n2,n3){}
\drawedge[ELdist=1.7](n1,n2){}
\drawedge[curvedepth=3](n-1,n){}
\drawedge(n-1,n){}
\drawedge[curvedepth=-3](n-1,n){}
\drawedge[curvedepth=3](n3,ndots){}
\drawedge(n3,ndots){}
\drawedge[curvedepth=-3](n3,ndots){}
\drawedge[curvedepth=3](ndots,n-1){}
\drawedge(ndots,n-1){}
\drawedge[curvedepth=-3](ndots,n-1){}
\drawloop[ELdist=1.5,loopangle=-90](n1){}
\drawloop[ELdist=1.5,loopangle=-90, loopdiam=7](n1){}
\end{picture}
\end{center}
\caption{The digraphs $\mathcal{G}_{n,2}$ and $\mathcal{G}_{n,3}$.}
\label{fig:half}
\end{figure}
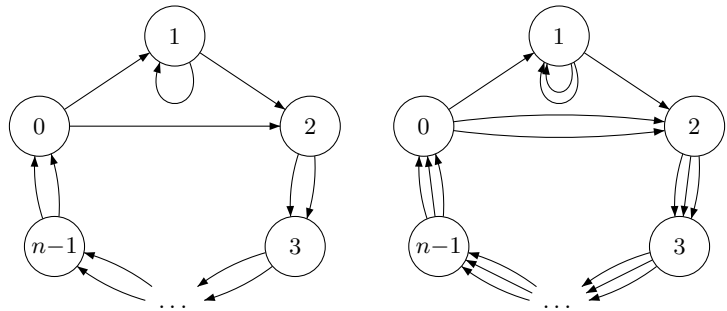

We generalize the above result to obtain digraphs with different values of the synchronizing ratio.

\begin{theorem}
For every integers $d \geq 1$ and $n \ge 3d$ there is a $k$-out-regular digraph with $n$ vertices and the synchronizing ratio $1-\frac{1}{k^d}$.
\end{theorem}
\begin{proof}
We will define a digraph $\mathcal{H}^d_{n,k}$ as follows; see Figure~\ref{fig:synch_frac}. The set of vertices $V$ is $\{0,1, \ldots n-1\}$. There are edges $(i,i+1)$ of multiplicity $k$ for every $2d \leq i \leq n-1$. There are edges $(i,i+1)$ of multiplicity $k-1$ for every $0 \leq i < 2d$.
For every $0 \leq i < d$ the vertex $2i+1$ has a loop.
The remaining edges of multiplicity 1 are of the form $(2i, 2i+2)$ for every $0 \leq i < d$.

Consider now an arbitrary coloring $\mathrsfs{H}^d_{n,k}$ of the digraph $\mathcal{H}^d_{n,k}$.
Let $x_i$ be the letter of the edge $(2i, 2i+2)$ and $y_i$ be the letter of the loop $2i+1$, for $0 \leq i < d$.
We will show that the automaton $\mathrsfs{H}^d_{n,k}$ is synchronizing if and only if $x_i = y_i$ for every $i$.
It will immediately imply that the synchronizing ratio of $\mathcal{H}^d_{n,k}$ is equal to $1-\frac{1}{k^d}$.

If $x_i=y_i$ for every $i$ then every letter acts as a permutation; thus the automaton is not synchronizing.
Assume now that there is $\ell$ such $x_\ell \neq y_\ell$, and let $\ell$ be the smallest integer with this property.
In order to show that $\mathrsfs{H}^d_{n,k}$ is synchronizing it remains to prove that any pair can be synchronized. Let $p$ and $q$ be a pair of states. Since the automaton is strongly connected, there is a word $u$ mapping $p$ to $2d$. Let $q' = \delta(q,u)$.
Let $v$ be a shortest word mapping $q'$ to a state in $D=\{i \in V \mid i\geq 2d\} \cup \{0\}$. Note that $|v| \le d$.
Now we have $\delta(q,uv) \in D$, and also $\delta(p,uv) = \delta(2d,v) \in D$, because $n \ge 3d$.
By the fact that $\ell$ is the smallest integer with the property $x_\ell \neq y_\ell$ we obtain $\delta(s',y^n_\ell) = \delta(t',y^n_\ell)$, which concludes the proof.
\qed
\end{proof}

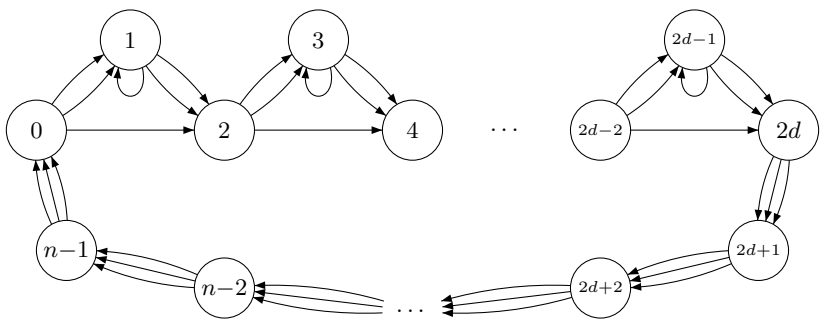
\begin{figure}[th]
\begin{center}
\unitlength .5mm
\begin{picture}(200,90)(0,-90)
\gasset{Nw=16,Nh=16,Nmr=8,loopdiam=7}
\node(n1)(25,-16.0){1}
\node(n)(0,-40.0){$0$}
\node(n2)(50.0,-40.0){2}
\node(n31)(50.0,-82.0){$n{-}2$}
\node(n3)(75,-16){3}
\node(n4)(100,-40.0){4}
\node(n6)(150,-40.0){$\scriptstyle 2d-2$}
\node(n7)(175,-16.0){$\scriptstyle 2d-1$}
\node(n8)(200,-40.0){$2d$}
\node(n9)(192,-72){$\scriptstyle 2d+1$}
\node(n10)(150,-82.0){$\scriptstyle 2d+2$}
\node(n-1)(8.0,-72.0){$n{-}1$}
\node[Nframe=n](ndots)(100,-88){$\dots$}
\node[Nframe=n](ndots2)(125,-40){$\dots$}
\drawedge(n,n2){}
\drawedge[curvedepth=3](n,n1){}
\drawedge[curvedepth=-3](n,n1){}
\drawedge[curvedepth=3](n1,n2){}
\drawedge[curvedepth=-3](n1,n2){}
\drawedge[curvedepth=3](n-1,n){}
\drawedge(n-1,n){}
\drawedge[curvedepth=-3](n-1,n){}
\drawedge[curvedepth=3](ndots,n31){}
\drawedge(ndots,n31){}
\drawedge[curvedepth=-3](ndots,n31){}
\drawedge[curvedepth=3](n31,n-1){}
\drawedge(n31,n-1){}
\drawedge[curvedepth=-3](n31,n-1){}
\drawloop[ELdist=1.5,loopangle=-90](n1){}
\drawloop[ELdist=1.5,loopangle=-90](n3){}
\drawloop[ELdist=1.5,loopangle=-90](n7){}
\drawedge[curvedepth=3](n2,n3){}
\drawedge[curvedepth=-3](n2,n3){}
\drawedge[curvedepth=3](n3,n4){}
\drawedge[curvedepth=-3](n3,n4){}
\drawedge[curvedepth=3](n6,n7){}
\drawedge[curvedepth=-3](n6,n7){}
\drawedge[curvedepth=3](n7,n8){}
\drawedge[curvedepth=-3](n7,n8){}
\drawedge[curvedepth=3](n8,n9){}
\drawedge(n8,n9){}
\drawedge[curvedepth=-3](n8,n9){}
\drawedge[curvedepth=3](n9,n10){}
\drawedge(n9,n10){}
\drawedge[curvedepth=-3](n9,n10){}
\drawedge[curvedepth=3](n10,ndots){}
\drawedge(n10,ndots){}
\drawedge[curvedepth=-3](n10,ndots){}
\drawedge(n2,n4){}
\drawedge(n6,n8){}
\end{picture}
\end{center}
\caption{The digraph $\mathcal{H}^d_{n,3}$.}
\label{fig:synch_frac}
\end{figure}

\begin{remark}
There exist many other digraphs with synchronizing ratio $1-\frac{1}{k^d}$.
Note that we can replace the $k$-path $(2d,\ldots,n-1,0)$ of $\mathcal{H}^d_{n,k}$ with any acyclic multigraph such that:
any vertex is reachable from the vertex $2d$; from any vertex we can reach the vertex $0$; and every path from $2d$ to $0$ is of the same length $\ge 3d$.
\end{remark}

\section{Conclusions and open problems}

In this section we summarize all the conjectures and open problems.
All of the conjectures are supported by experimental evidence.

\begin{conjecture}
The minimum value of the synchronizing ratio among all $k$-out-regular digraphs with $n$ vertices is equal to $\frac{k-1}{k}$, except for the case $k=2$ and $n=6$ when it is equal to $\frac{30}{64}$.
\end{conjecture}

The conjecture was verified for small values of $n$ and $k$ (see Table~\ref{tab:exhaustive_results_1}).
It implies that a uniformly random coloring of a primitive strongly connected digraph is synchronizing with probability at least $1/2$, and hence it would justify the algorithm finding synchronizing coloring randomly.

To state the next conjecture, let say that a $\emph{gap}$ in the distribution of the number of synchronizing colorings is a maximal interval of integers divisible by $k!$, such that there are no digraphs with the number of synchronizing colorings in this interval.
Thus the conjecture above states that for every $k$ and $n \neq 6$ large enough there is the gap $[k!,\frac{k-1}{k}(k!)^n-k!]$.

\begin{conjecture}
For every $k$ and $g \ge 1$, there is an $n$ large enough such that there are at least $g$ gaps in the distribution of the number of synchronizing colorings of $k$-out-regular digraphs with $n$ vertices.
\end{conjecture}

The following conjecture can be stated either in the class of strongly connected and aperiodic digraphs, or in the class of all digraphs.

\begin{conjecture}
For every $k \ge 2$, the fraction of totally synchronizing digraphs among all $k$-out-regular digraphs with $n$ vertices tends to $1$ as $n$ goes to infinity.
\end{conjecture}
A recent non-trivial theorem states that a random automaton is synchronizing with high probability~\cite{Berlinkov2013OnTheProbabilityToBeSynchronizable,Nicaud2014FastSynchronizationOfRandomAutomata}. Conjecture 3 can be seen as a further development of this statement.

We conclude with the following problem related to computing the number of synchronizing colorings.

\begin{problem}
Given a $k$-out-regular digraph $\mathcal{G}$ with $n$ vertices, what is the computational complexity of checking whether $\mathcal{G}$ is totally synchronizing.
\end{problem}

\subsubsection{Acknowledgment}
The authors want to thank Mikhail Volkov for his significant contributions to the theory of synchronizing automata on the occasion of his $60^{\text{th}}$ birthday.

\bibliographystyle{plain}

\end{document}